\documentclass[10pt,sigplan,anonymous=false,balance=true]{acmart}
\renewcommand\footnotetextcopyrightpermission[1]{}
\acmConference[IWLS'21]{}{July 19-22, 2021}{Virtual Conference}

\usepackage{tikz}
\usepackage{cleveref}
\usepackage{subcaption}
\usepackage{microtype}
\usetikzlibrary{calc,shapes,positioning}

\newtheorem{theorem}{Theorem}
\newtheorem{lemma}[theorem]{Lemma}
\newtheorem{claim}[theorem]{Claim}

\newcommand{\FOT}{\textsc{FOT}}
\newcommand{\FO}{\textsc{FO}}
\newcommand{\FI}{\textsc{FI}}
\newcommand{\D}{d}
\newcommand{\RD}{\mathit{rd}}

\usepackage[linesnumbered,noend,noline]{algorithm2e}
\newcommand{\assign}{$~\leftarrow~$}

\begin{document}

\title{Irredundant Buffer and Splitter Insertion and Scheduling-Based Optimization for AQFP Circuits}

\author{Siang-Yun Lee}
\affiliation{
  \institution{EPFL, Switzerland}
}

\author{Heinz Riener}
\affiliation{
  \institution{EPFL, Switzerland}
}

\author{Giovanni De Micheli}
\affiliation{
  \institution{EPFL, Switzerland}
}

\begin{abstract}
The adiabatic quantum-flux parametron (AQFP) is a promising energy-efficient superconducting technology. Before technology mapping, additional buffer and splitter cells need to be inserted into AQFP circuits to fulfill two special constraints: (1)~Input signals to a logic gate need to arrive at the same time, thus shorter paths need to be delayed with buffers. (2)~The output signal of a logic gate has to be actively branched with splitters if it drives multiple fanouts. Buffers and splitters largely increase the area and delay in AQFP circuits. Na\"ive buffer and splitter insertion and light-weight optimization using retiming techniques have been used in related works, and it is not clear how much space there is for further optimization.  In this paper, we develop (a) a linear-time algorithm to insert buffers and splitters irredundantly, and (b) optimization methods by scheduling and by moving groups of gates, called chunks, together. Experimental results show a reduction of up to $39\%$ on buffer and splitter cost. Moreover, as the technology is still developing and assumptions on the physical constraints are not clear yet, we also discuss the impacts of different assumptions with experimental results to motivate future research on AQFP register design.
\end{abstract}
\maketitle

\section{Introduction}\label{sec:intro}
Superconducting electronics is an emerging domain arising from the demand for ultra-low power consumption.
Among various superconducting logic families, the \emph{adiabatic quantum-flux parametron}~(AQFP)~\cite{takeuchi13aqfp} is a technology featuring zero static energy consumption and very small switching energy dissipation.
Two of the challenges in AQFP circuit design come from the \emph{path-balancing} and \emph{fanout-branching} requirements which are not needed in traditional CMOS logic circuits.

\textit{Path-balancing:} The AQFP gates are AC-biased. Each AQFP gate receives an alternating excitation current to periodically release its output signal and reset its state. All AQFP clocking schemes~\cite{takeuchi17fourphase,takeuchi19delayline} require that the input signals of a logic gate be released at the previous clocking phase. In other words, all data paths must be of the same length. Whereas shortening longer paths is not always possible, \emph{buffers} need to be inserted to delay shorter paths. 

\textit{Fanout-branching:} In the AQFP technology, logical~$0$ and~$1$ are represented with different current directions. As the output current of an AQFP gate is limited, it has to be amplified by a \emph{splitter} before branching into multiple fanouts. AQFP splitters are also clocked.

As the research at the physical level rapidly develops and the fabrication capability grows for larger and more complex circuits, design automation tools specialized for AQFP are in need. Pioneering works attempt to adapt existing tools to fulfill the path-balancing and fanout-branching constraints with post-synthesis modifications and optimization.
In~\cite{ayala2020semi} and~\cite{cai2019buffer}, after classical logic synthesis, path-balancing buffers and fanout-branching splitters are inserted separately and then retiming-like algorithms are applied to reduce the buffer and splitter cost locally. A majority-based logic synthesis flow considering AQFP buffer and splitter costs is proposed in~\cite{testa2021algebraic}, which emphasizes on reducing circuit depth and restricting the increase of fanout count. In~\cite{marakkalage2021}, consideration of the balancing and branching constraints is integrated in exact-synthesis-based rewriting.
However, in the results of these works, buffers and splitters still make up for over $50\%$, and sometimes up to $80\%$, of the total cost. 
Moreover, as the technology is still being developed, the physical constraints are ever-changing and assumptions on the requirements vary across different works and are often unclear, which makes them difficult to compare with. For example, whether primary inputs and primary outputs need to be path-balanced and fanout-branched depends on the design of interfacing registers, which is still under development~\cite{saito2021logic}.

While the path-balancing constraint also exists for the \emph{rapid single-flux-quantum} (RSFQ) technology and optimization methods have been researched~\cite{pasandi2018pbmap}, an important distinction is that RSFQ splitters are not clocked but AQFP splitters are. For this reason, fanout-branching has to be considered together with path-balancing in AQFP, which makes the problem more complicated. If only path balancing needs to be done, as in RSFQ, the optimal way of inserting buffers without changing the logic structure can be found in linear time. However, buffer and splitter insertion in AQFP is non-trivial even without logic optimization.
Thus, we limit our investigation to the problem of AQFP buffer and splitter insertion without logic transformation. 

In this paper, observations about the complexity of the defined problem and systematic methods to deal with it are presented, and the impact of the technology assumptions are experimented and discussed.
In \Cref{sec:buf-count}, a linear-time algorithm to count irredundant buffers is presented, which subsumes the retiming and optimization techniques proposed in~\cite{cai2019buffer}. Different than the previous work, we consider buffers and splitters together and count them in an irredundant way such that the ``optimizations'' in~\cite{cai2019buffer} are considered without extra effort. Also, we observe that the irredundant construction is not optimal yet. On top of the locally irredundant buffer and splitter insertion, efforts have to be made in finding a suitable depth assignment to logic gates to achieve the global optimum. In \Cref{sec:optimize}, methods to obtain an initial depth assignment and to adjust it incrementally are presented. To escape from local minimum, we propose to move groups of gates together as a chunk. Experimental results show that obtaining better depth assignments using the proposed methods reduces the number of buffers by up to $39\%$. Moreover, various possible technology assumptions are first discussed in \Cref{sec:assumptions} and then experimented in \Cref{sec:exp}. The results suggest that branching and balancing of primary inputs have greater impacts on the buffer count of about $50\%$ and $30\%$, respectively.

\section{Background}\label{sec:background}
\subsection{Adiabatic Quantum-Flux Parametron}\label{subsec:aqfp}
The \emph{adiabatic quantum-flux parametron}~(AQFP) is an emerging superconducting technology shown to achieve promising energy efficiency.~\cite{takeuchi13aqfp} The basic circuit components in this technology are the buffer cell and the branch cell. The majority-$3$ logic gate can be constructed by combining three buffer cells with a $3$-to-$1$ branch cell, from which other logic gates, such as the AND gate and the OR gate, can be built with constant cells (biased buffer cells). Input negation of logic gates is realized using a negative mutual inductance and is of no extra cost.~\cite{takeuchi15library} The commonly-used cost metric of AQFP circuits is the \emph{Josephson junction}~(JJ) count. A buffer costs $2$ JJs and a majority-$3$ gate costs $6$ JJs.

Logic gates in an AQFP circuit need to be activated and deactivated periodically by an excitation current.~\cite{takeuchi17fourphase} In other words, every gate in an AQFP circuit is clocked, and all input signals have to arrive at the same clock cycle. To ensure this, shorter data paths need to be delayed with clocked buffers. Moreover, the output signal of AQFP logic gates cannot be directly branched to feed into multiple fanouts. Instead, splitters are placed at the output of multi-fanout gates to amplify the output current. A splitter cell is composed of a buffer cell and a $1$-to-$n$ branch cell (usually, $2 \leq n \leq 4$) and is also clocked. As the cost of splitters comes mostly from the buffer cells, in the remaining of this paper, we do not distinguish buffers and splitters and will model them with the same abstract data structure.

\subsection{Terminology}\label{subsec:terminology}
A \emph{(logic) network} is a directed acyclic graph defined by a pair $(V, E)$ of a set $V$ of nodes and a set $E$ of directed edges. The node set $V = I \cup O \cup G$ is disjointly composed of a set $I$ of \emph{primary inputs} (PIs), a set $O$ of \emph{primary outputs} (POs), and a set $G$ of \emph{(logic) gates}. Each PI has in-degree $0$ and unbounded out-degree, whereas each PO has in-degree $1$ and out-degree $0$. The out-degree of each gate is unbounded and the in-degree is a fixed number depending on the type of the gate. 
For any gate $g \in G$, the \emph{fanins} of $g$, denoted as $\FI(g)$, is the set of gates and PIs connected to $g$ with an incoming edge. Similarly, the \emph{fanouts} of $g$, denoted as $\FO(g)$, is the set of gates and POs connected to $g$ with an outgoing edge. Fanouts are also defined for PIs.

A \emph{mapped network} $N^\prime$ is a network whose node set $V^\prime$ is extended with a set $B$ of \emph{buffers}. A buffer is a node with in-degree $1$. In a mapped network, the definition of the fanouts of a gate is modified by ignoring any intermediate buffers, i.e., a path from a gate $g$ to one of its fanouts $g_o \in \FO(g) \subset (G \cup O)$ may include any number of buffers in $B$, but never another gate in $G - \{g, g_o\}$. The definition of fanins is modified similarly.
The \emph{fanout tree} of a gate $g$, denoted by $\FOT(g)$, is the set of buffers between $g$ and any gate or PO in $\FO(g)$. Fanout trees are also defined for PIs.

For each node $n$ in a network, the \emph{depth} of $n$, denoted by $\D(n)$, is a non-negative integer assigned to $n$. The depth of a network $N = (V = I \cup O \cup G, E)$ is defined as
\begin{equation}
    \D(N) = \max\limits_{o \in O} \D(o).
\end{equation} 
Moreover, the \emph{relative depth} between a PI or a gate $n \in (I \cup G)$ and one of its fanout $n_o \in \FO(n) \subset (G \cup O)$, is denoted and defined as 
\begin{equation}
    \RD(n, n_o) = \D(n_o)-\D(n).
\end{equation} 
Note that relative depth is only defined among PIs, gates, and POs.

\section{Technology Assumptions}\label{sec:assumptions}
To fulfill the needs in the AQFP technology for fanout-branching and path-balancing, we define the following two properties for a mapped network $N^\prime = (V^\prime = I \cup O \cup G \cup B, E^\prime)$ with a depth assignment. Given the \emph{splitting capacities} $s_i, s_g, s_b$ of each type of node, 
\begin{enumerate}
    \item $N^\prime$ is \emph{path-balanced} if 
    \begin{align}
        &\forall n_1, n_2 \in V^\prime : (n_1, n_2) \in E^\prime \Rightarrow \D(n_1) = \D(n_2) - 1\text{, } \\
        &\forall i \in I : \D(i) = 0\text{, and}\label{eqn:balance-pi} \\
        &\forall o \in O : \D(o) = \D(N^\prime)\label{eqn:balance-po}.
    \end{align}
    \item $N^\prime$ is \emph{properly-branched} if every PI has an out-degree no larger than $s_i = 1$, every gate has an out-degree no larger than $s_g = 1$, and every buffer has an out-degree no larger than $s_b$.
\end{enumerate}

An (unmapped) network $N$ with a depth assignment is said to be \emph{legal} if a path-balanced and properly-branched mapped network $N^\prime$ can be extended from $N$.

Logic networks defined in \Cref{subsec:terminology} model the combinational parts of digital circuits. In practice, PIs of a network are usually provided by the register outputs of the previous sequential stage and POs are connected to the register inputs of the next stage. Depending on how the registers are implemented, different assumptions on whether PIs and POs need to be path-balanced or branched may arise. 

\subsection{Path-Balancing of PIs} 
It is possible to design registers that can hold and output its value at every clock cycle. In this case, the PI nodes in our model can be placed at any depth, i.e., condition~\ref{eqn:balance-pi} is removed.

\subsection{Path-Balancing of POs} 
In most related works, it is assumed that all PO signals must arrive at the register inputs at the same clock cycle. That is, POs are path-balanced to ensure robust operations. If the PI values are always available and stable until the next register update, shorter paths from PI to PO simply compute the same result repeatedly in the later cycles when longer paths are still computing. In this case, shorter paths do not have to be aligned with the longest path (the critical path). In other words, the PO nodes in our model are no longer limited to be placed at the same depth, i.e., condition~\ref{eqn:balance-po} is removed.
However, there may still be constraints on the PO depths depending on the clocking scheme used. For example, a $4$-phase clocking scheme~\cite{takeuchi17fourphase} may require that the depths of PO nodes must be a multiple of $4$ because the registers can only take inputs in one of the four clock phases.

\subsection{Branching of PIs}
When a register drives multiple outputs, we may or may not need to insert splitters to ensure a large enough current, depending on the physical implementation of the register. If the registers are capable of producing large current, $s_i$ can be set to infinity. Otherwise, it is also possible to duplicate the frequently-used PIs in the register file to avoid deep splitter trees, or to design special large-capacity buffers having a higher $s_b$ value and use them for PIs with many fanouts.

\subsection{Branching and Inversion of POs}
If a gate output feeds into multiple registers, then splitters are always needed. If the negated output of a majority gate is required by the next sequential stage, we can push the output inversion to the gate's inputs because the majority function is self-dual~\cite{Muroga61} and input negation is for free in AQFP. However, if a gate output is needed by the next stage once in the regular form and once in the negated form, then we not only need a splitter, but also an additional NOT gate made of an input-negated buffer.

\subsection{Problem Formulation}\label{subsec:problem}
In this paper, we focus on the problem of AQFP buffer insertion after logic synthesis without changing the structure of the original network, formulated as follows:

Given a network $N = (V = I \cup O \cup G, E)$ and the value of the parameter $s_b$, find a mapped network $N^\prime = (V^\prime = I \cup O \cup G \cup B, E^\prime)$, such that:
\begin{enumerate}
    \item $N^\prime$ is path-balanced and properly-branched.
    \item For all gates $g \in G$, $\FO(g)$ and $\FI(g)$ remain the same in $N^\prime$ as in $N$.
    \item $|V^\prime|$ is minimized. Since $V^\prime = V \cup B$, it is equivalent to $|B|$ being minimized.
\end{enumerate}
We call such $N^\prime$ a \emph{minimum} mapped network for $N$.

\section{Irredundant Buffer Insertion}\label{sec:buf-count}
A mapped network is said to be \emph{irredundant} if the following two conditions hold.
\begin{enumerate}
    \item There is no dangling buffer, i.e., every buffer has at least one outgoing edge.
    \item There does not exist any pair of two buffers whose incoming edges are connected from the same node and both of them have out-degrees smaller than $s_b$.
\end{enumerate}

We consider only irredundant networks in the remaining of this paper. In this section, we will explain how the problem formulated in \Cref{subsec:problem} can be approached, starting from the following observation.

\begin{claim}\label{claim: assign depth}
Given a network $N=(V, E)$, finding a minimum mapped network for $N$ is essentially finding a depth assignment to every node in $V$.
\end{claim}

To show why \Cref{claim: assign depth} is true, we will first formulate \Cref{lemma: minimize sum of fot} to show that the buffer set in an irredundant mapped network can be decomposed into fanout trees of each gate. Then, we will present Algorithm~\ref{alg:opt-fot} to show how the irredundant fanout tree of a gate $g$ can be constructed given the relative depths of its fanouts. Thus, once a depth assignment is given, the total size of fanout trees is decided, so as the size of the mapped network.

\begin{lemma}\label{lemma: minimize sum of fot}
In any irredundant mapped network with PI set~$I$, gate set~$G$, and buffer set~$B$, 
\begin{equation*}
    B = \bigcup\limits_{g \in G} \FOT(g) \cup \bigcup\limits_{i \in I} \FOT(i).
\end{equation*}
\end{lemma}
\begin{proof}
By definition, a buffer has exactly one incoming edge.
The adjacent node connected to a buffer with its incoming edge is either another buffer in $B$, a gate in $G$, or an PI in $I$ because POs have no outgoing edge.
Going from a buffer $b$ in the opposite direction of edges and continue tracing until a gate $g$ or a PI $i$ is met, we have $b \in \FOT(g)$ (or $b \in \FOT(i)$) because there is no dangling buffer tree (rule~1 for irredundant networks). 
Hence, for each buffer $b \in B$, there is either a gate $g \in G$ such that $b \in \FOT(g)$, or there is a PI $i \in I$ such that $b \in \FOT(i)$. Moreover, this gate or PI is unique for each $b$. For each gate $g \in G$ and for each PI $i \in I$, $\FOT(g) \subseteq B$ and $\FOT(i) \subseteq B$ by definition. Thus, the set of non-empty fanout trees is a partitioning of $B$.
\end{proof}

\begin{algorithm}
\begin{small}
    \DontPrintSemicolon
    \KwIn{A gate $g$}
    \KwOut{The size $|\FOT(g)|$ of the fanout tree of $g$}
    \it
    $l_{max}$ \assign $\max\limits_{g_o \in \FO(g)}\RD(g, g_o)$\;
    count \assign $0$\;
    edges \assign $|\{g_o \in \FO(g) : \RD(g, g_o) = l_{max}\}|$\;
    \For{$l = l_{max} - 1 ~\text{\bf downto} ~1$}{
        buffers \assign $\lceil \frac{\text{edges}}{s_b} \rceil$\;
        count \assign count $+$ buffers\;
        edges \assign buffers $ + ~|\{g_o \in \FO(g) : \RD(g, g_o) = l\}|$\;
    }
    {\bf assert} edges $ = 1$\;
    \KwRet{count}
    \vspace{.2cm}
    \caption{Irredundant fanout tree construction given relative depths of fanouts.}
    \label{alg:opt-fot}
\end{small}
\end{algorithm}

For any gate $g$, given relative depths $\RD(g, g_o)$ of its fanouts $g_o \in \FO(g)$, the size of its fanout tree $|FOT(g)|$ can be computed with Algorithm~\ref{alg:opt-fot}.
The algorithm iterates over all levels $l$ from the relative depth of the highest fanout down to $1$, and counts the number of buffers (variable \textit{buffers}) needed at each level. The total number of buffers is accumulated in variable \textit{count} (line $6$). At each level $l$, variable \textit{edges} keeps the number of edges ending in some node of relative depth $l$, which is simply the number of buffers and fanouts at this level (line $7$). Then, the number of buffers needed at the lower level $l-1$ is computed by the number of edges starting at $l-1$ (i.e., the number of edges ending at $l$) divided by the splitting capacity $s_b$ and rounded up (line $5$). This algorithm works also for constructing the fanout tree of a given PI. If the fanout information is stored in a data structure that the size $|\{g_o \in \FO(g) : \RD(g, g_o) = l\}|$ for any given value $l$ can be queried in constant time, then the algorithm runs in linear time with respect to $|\FO(g)|$.

\Cref{fig:alg1} illustrates an example execution of Algorithm~\ref{alg:opt-fot}, where circles are gates and squares are buffers, and $s_b=2$ is assumed. The concerned gate $g$ has one fanout of relative depth $2$ and three fanouts of relative depth $5$. The total number of buffers in the fanout tree is $5$.

\begin{figure}[t]
\center
\begin{tikzpicture}[>=latex]
  \tikzstyle{square} = [draw,minimum size=0.2cm]
  \tikzstyle{bubble} = [draw,circle,minimum size=0.35cm]
  \matrix[row sep=0.4cm,column sep=0.15cm, nodes={}] {
    \node (layer5) {};   \node[bubble] (7) {}; &                       & \node[bubble] (8) {}; & \node[bubble] (9) {}; & \node (comp5) {}; \\
    \node (layer4) {};                         & \node[square] (5) {}; &                       & \node[square] (6) {}; & \node (comp4) {}; \\
    \node (layer3) {};                         &                       & \node[square] (4) {}; &                       & \node (comp3) {}; \\
    \node (layer2) {};   \node[bubble] (2) {}; &                       & \node[square] (3) {}; &                       & \node (comp2) {}; \\
    \node (layer1) {};                         & \node[square] (1) {}; &                       &                       & \node (comp1) {}; \\
    \node (layer0) {};                         & \node[bubble] (0) {}; &                       &                       & \node (comp0) {}; \\
  };

  \draw[->] (0) -- (1); 
  \draw[->] (1) -- (2); 
  \draw[->] (1) -- (3);
  \draw[->] (3) -- (4); 
  \draw[->] (4) -- (5); 
  \draw[->] (4) -- (6);
  \draw[->] (5) -- (7); 
  \draw[->] (5) -- (8);
  \draw[->] (6) -- (9);

  \node at (0) {\footnotesize $g$};
  \node[left of=layer1] {\footnotesize $l = 1$};
  \node[left of=layer2] {\footnotesize $l = 2$};
  \node[left of=layer3] {\footnotesize $l = 3$};
  \node[left of=layer4] {\footnotesize $l = 4$};
  \node[left of=layer5] {\footnotesize $l = 5$};

  \node[right=-0.4cm of comp1] {\footnotesize \it buffers $ = \lceil\frac{2}{2}\rceil = 1$,};
  \node[right=-0.4cm of comp2] {\footnotesize \it buffers $ = \lceil\frac{1}{2}\rceil = 1$,};
  \node[right=-0.4cm of comp3] {\footnotesize \it buffers $ = \lceil\frac{2}{2}\rceil = 1$,};
  \node[right=-0.4cm of comp4] {\footnotesize \it buffers $ = \lceil\frac{3}{2}\rceil = 2$,};

  \node[right=2.2cm of comp1] {\footnotesize \it edges $ = 1$};
  \node[right=2.2cm of comp2] {\footnotesize \it edges $ = 2$};
  \node[right=2.2cm of comp3] {\footnotesize \it edges $ = 1$};
  \node[right=2.2cm of comp4] {\footnotesize \it edges $ = 2$};
  \node[right=2.2cm of comp5] {\footnotesize \it edges $ = 3$};
\end{tikzpicture}
\caption{Example sub-network to illustrate Algorithm 1.}
\vspace{-0.5em}
\label{fig:alg1}
\end{figure}
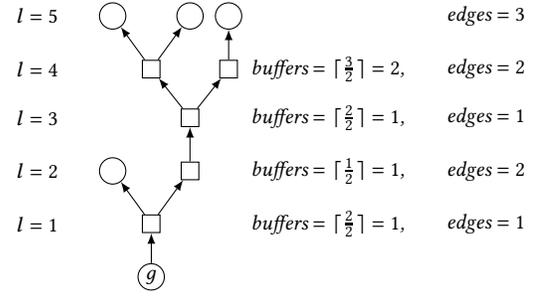

The constructed fanout tree is guaranteed to be irredundant because only the minimum number of buffers is inserted at each level based on the number of outgoing edges needed. Note that the retiming optimization proposed in~\cite{cai2019buffer}, which pushes buffers from the outputs of a splitter to its input, is already considered during construction of irredundant fanout trees.

Moreover, Algorithm~\ref{alg:opt-fot} also verifies whether it is possible to build a properly-branched network with the given depth assignment. In line $8$, the assertion makes sure that the gate $g$ has only one outgoing edge. 
Running the algorithm for all PIs and gates in a depth-assigned network, by \Cref{lemma: minimize sum of fot}, a mapped network is derived. The mapped network is guaranteed to be properly-branched if the assertion in line $8$ never fails. It is also path-balanced because every node is connected to a node at exactly one level lower. As Algorithm~\ref{alg:opt-fot} is deterministic, we conclude that the number of irredundant buffers for a given depth assignment is unique. 

\section{Optimization on Depth Assignment}\label{sec:optimize}
Following \Cref{claim: assign depth}, in this section, we attempt to find a good depth assignment to minimize the total number of buffers in the mapped network. In \Cref{subsec:scheduling}, we first obtain an initial depth assignment using scheduling algorithms. Then, in \Cref{subsec:chunk}, we try to move gates up or down to reduce the total number of buffers. \emph{Moving} a gate $g$ \emph{up} by $l$ levels means that $\D(g)$ is increased by $l$ while the depths of the other gates remain the same. Similarly, moving $g$ \emph{down} means $\D(g)$ is decreased. During the entire process, we always ensure that the network is legal. 

\subsection{Obtaining an Initial Depth Assignment}\label{subsec:scheduling}
An initial depth assignment can be obtained using an \emph{as-soon-as-possible scheduling}~(ASAP) algorithm which assigns the smallest possible depth to each gate. To ensure that the network can be path-balanced and properly-branched after mapping, enough depths for a balanced fanout tree are reserved at the output of every multi-fanout gates, which is calculated as
\begin{equation}
    \left\lceil \frac{\log (|\FO(g)|)}{\log (s_b)} \right\rceil.
\end{equation}
Then, an \emph{as-late-as-possible scheduling}~(ALAP) can be applied using an upper bound $\D(N)$ obtained by ASAP.

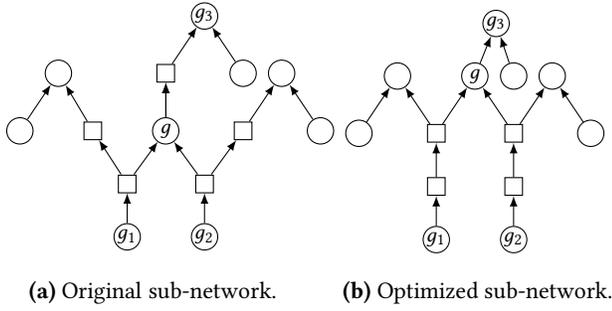
\begin{figure}[t]
\center
\begin{subfigure}[b]{0.49\columnwidth}
    \centering
    \begin{tikzpicture}[>=latex, baseline=(13)]
  \tikzstyle{s} = [draw,minimum size=0.2cm]
  \tikzstyle{b} = [draw,circle,minimum size=0.35cm]
  \matrix[row sep=0.4cm,column sep=0.15cm, nodes={}] {
                     &                  &                  &                  &                  & \node[b](13) {}; &                  &                  &                  \\
                     & \node[b] (9) {}; &                  &                  & \node[s](10) {}; &                  & \node[b](11) {}; & \node[b](12) {}; &                  \\
    \node[b] (4) {}; &                  & \node[s] (5) {}; &                  & \node[b] (6) {}; &                  & \node[s] (7) {}; &                  & \node[b] (8) {}; \\
                     &                  &                  & \node[s] (2) {}; &                  & \node[s] (3) {}; &                  &                  &                  \\
                     &                  &                  & \node[b] (0) {}; &                  & \node[b] (1) {}; &                  &                  &                  \\
  };

  \draw[->] (0) -- (2); 
  \draw[->] (1) -- (3);
  \draw[->] (2) -- (5); 
  \draw[->] (2) -- (6);
  \draw[->] (3) -- (6);
  \draw[->] (3) -- (7);
  \draw[->] (4) -- (9); 
  \draw[->] (5) -- (9);
  \draw[->] (6) -- (10);
  \draw[->] (6) -- (10);
  \draw[->] (7) -- (12);
  \draw[->] (8) -- (12);
  \draw[->] (10) -- (13); 
  \draw[->] (11) -- (13);
 
  \node at (0)  {\footnotesize $g_1$};
  \node at (1)  {\footnotesize $g_2$};
  \node at (6)  {\footnotesize $g$};
  \node at (13) {\footnotesize $g_3$};
\end{tikzpicture}
    \caption{Original sub-network.}\label{fig:ex1}
  \end{subfigure}
  \hfill
  \begin{subfigure}[b]{0.49\columnwidth}
    \centering
    \begin{tikzpicture}[>=latex, baseline=(12shift)]
  \tikzstyle{s} = [draw,minimum size=0.2cm]
  \tikzstyle{b} = [draw,circle,minimum size=0.35cm]
  \matrix[row sep=0.4cm,column sep=0.15cm, nodes={}] {
                     &                  &                  & \node   (12) {}; &                  &                  &                  \\
                     & \node[b] (8) {}; &                  & \node[b] (9) {}; & \node[b](11) {}; & \node[b](10) {}; &                  \\
    \node[b] (4) {}; &                  & \node[s] (5) {}; &                  & \node[s] (6) {}; &                  & \node[b] (7) {}; \\
                     &                  & \node[s] (2) {}; &                  & \node[s] (3) {}; &                  &                  \\
                     &                  & \node[b] (0) {}; &                  & \node[b] (1) {}; &                  &                  \\
  };

  \node[b] (12shift) at ([xshift=8pt] 12) {};
  
  \draw[->] (0) -- (2); 
  \draw[->] (1) -- (3);
  \draw[->] (2) -- (5); 
  \draw[->] (3) -- (6);
  \draw[->] (4) -- (8); 
  \draw[->] (5) -- (8);
  \draw[->] (5) -- (9);
  \draw[->] (6) -- (9);
  \draw[->] (6) -- (10);
  \draw[->] (7) -- (10);
  \draw[->] (9)  -- (12shift); 
  \draw[->] (11) -- (12shift);
 
  \node at (0)  {\footnotesize $g_1$};
  \node at (1)  {\footnotesize $g_2$};
  \node at (9)  {\footnotesize $g$};
  \node at (12shift) {\footnotesize $g_3$};
\end{tikzpicture}
    \caption{Optimized sub-network.}\label{fig:ex2}
  \end{subfigure}
\caption{Example sub-network where ASAP does not lead to the optimum.}
\vspace{-0.5em}
\label{fig:ex2}
\end{figure}

However, neither ASAP nor ALAP achieves the global optimum. \Cref{fig:ex2}~(a) shows an example sub-network after ASAP, where circles are gates and squares are buffers. The gate $g$ is not the highest fanout of either of its fanins, thus moving $g$ up does not increase sizes of the fanout trees of $g_1$ and $g_2$. Moreover, the fanout $g_3$ is lower-bounded by its other fanin. Thus, by moving up $g$, as shown in \Cref{fig:ex2}~(b), a buffer is eliminated in its fanout tree.

The reason why this problem is not trivial is because a movement of a gate affects both its own fanout tree and its fanins' fanout trees. Moreover, in some cases, it is impossible to legally move a single gate and reduce the buffer count, but rearranging some gates altogether eventually leads to further reduction. Thus, in the following sections, groups of gates are identified and moved together as \emph{chunks}.

\subsection{Chunked Movement}\label{subsec:chunk}
A movement is \emph{legal} if the network remains legal after the movement. For example, if a gate $g$ has a fanout $g_o$ of relative depth $\RD(g, g_o) = 1$, then moving $g$ up alone is not legal. Similarly, if a gate $g$ has more than one fanouts, then moving any of its fanouts to $\D(g) + 1$ is not legal because there must be a buffer occupying the only outgoing edge of $g$ at $\D(g) + 1$. 

A pair of gates $(g, g_o): g_o \in \FO(g)$ are \emph{close} if either one of the following conditions holds:
\begin{enumerate}
    \item $\RD(g, g_o) = 1$, implying that $g_o$ is the only fanout of $g$.
    \item $|\FO(g)| > 1$ and $\RD(g, g_o) = 2$.
\end{enumerate}
If a gate $g$ and its fanout $g_o$ are not close, then there is \emph{flexibility} at the output of $g$ and at the input of $g_o$.

A \emph{chunk} $C$ is a set of closely-connected gates and can be seen as a super-node having multiple incoming and outgoing edges, called the \emph{input interfaces}~(IIs) and \emph{output interfaces}~(OIs), respectively. An interface is a pair $(g_c, g_f)$ of gates, where $g_c \in C$, $g_f \notin C$, and either $g_f \in \FI(g_c)$ (II) or $g_f \in \FO(g_c)$ (OI).

\begin{algorithm}
\begin{small}
    \DontPrintSemicolon
    \SetKw{cont}{continue}
    \SetKw{pop}{pop}
    \KwIn{An initial gate $g_0$}
    \KwOut{A chunk $C$ and its interfaces $T$}
    $C$ \assign $\{g_0\}$\;
    $F$ \assign $\{(g_0, g) : g \in \FI(g_0) \cup \FO(g_0)\}$\;
    $T$ \assign $\emptyset$\;
    \While{$F \neq \emptyset$}{
        $(g_c, g_f)$ \assign pop({$F$})\;
        \lIf{$g_f \in C$}{\cont}
        \If{$g_c$ and $g_f$ are close}{
            $C$ \assign $C \cup~g_f$\;
            $F$ \assign $F \cup~ \{(g_f, g) : g \in \FI(g_f) \cup \FO(g_f)\}$\;
        }
        \Else{
            $T$ \assign $T \cup \{(g_c, g_f)\}$\;
        }
    }
    \KwRet{$C, T$}
    \vspace{.2cm}
    \caption{Chunk construction.}
    \vspace{-0.5em}
    \label{alg:chunk}
\end{small}
\end{algorithm}

Algorithm~\ref{alg:chunk} illustrates how a chunk can be constructed. Starting from an initial gate $g_0$, a chunk is formed by exploring towards its fanins and fanouts and adding gates into the chunk if they are close (line $8$), or recording an input or output interface if there is flexibility (line $11$). When a new gate is added into the chunk, its fanins and fanouts are also explored (line $9$).

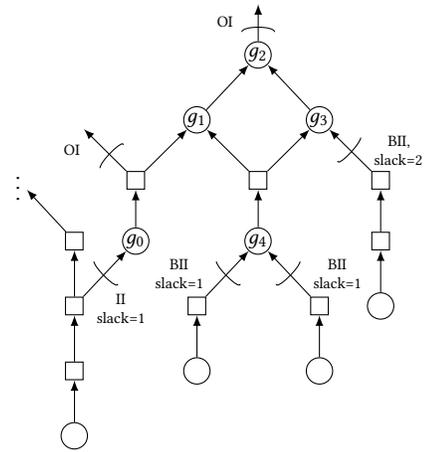
\begin{figure}[t]
\center
\begin{tikzpicture}[>=latex]
  \tikzstyle{s} = [draw,minimum size=0.2cm]
  \tikzstyle{b} = [draw,circle,minimum size=0.35cm]
  \matrix[row sep=0.5cm,column sep=0.45cm, nodes={}] {
                             &                  &                   &                   & \node[]  (e2) {}; &                   &                   \\
                             &                  &                   &                   & \node[b] (17) {}; &                   &                   \\
                             & \node[] (e1) {}; &                   & \node[b] (15) {}; &                   & \node[b] (16) {}; &                   \\ 
    \node[]  (e0){}; &                  & \node[s] (12) {}; &                   & \node[s] (13) {}; &                   & \node[s] (14) {}; \\        
                             & \node[s] (8) {}; & \node[b] (9)  {}; &                   & \node[b] (10) {}; &                   & \node[s] (11) {}; \\    
                             & \node[s] (4) {}; &                   & \node[s] (5)  {}; &                   & \node[s] (6) {};  & \node[b] (7)  {}; \\
                             & \node[s] (1) {}; &                   & \node[b] (2)  {}; &                   & \node[b] (3) {};  &                   \\
                             & \node[b] (0) {}; &                   &                   &                   &                   &                   \\
  };

  \draw[->] (0) -- (1); 

  \draw[->] (1) -- (4); 
  \draw[->] (2) -- (5); 
  \draw[->] (3) -- (6);

  \draw[->] (4) -- (8); 
  \draw[->] (4) edge
    node [below=4pt,xshift=6pt] {\begin{minipage}{0.75cm}\centering{}\tiny{}II\\slack=1\end{minipage}}
    node [sloped,anchor=south,auto=false,rotate=-90,yshift=-3pt,minimum width=12pt] (s0) {} (9);
  \draw[->] (5) edge
    node [left=4pt] {\begin{minipage}{0.75cm}\centering{}\tiny{}BII\\slack=1\end{minipage}}
    node [sloped,anchor=south,auto=false,rotate=-90,yshift=-3pt,minimum width=12pt] (s1) {} (10);
  \draw[->] (6) edge
    node [right=4pt] {\begin{minipage}{0.75cm}\centering{}\tiny{}BII\\slack=1\end{minipage}}
    node [sloped,anchor=south,auto=false,rotate=90,yshift=-3pt,minimum width=12pt] (s2) {} (10);
  \draw[->] (7) -- (11);

  \draw[->] (9)  -- (12); 
  \draw[->] (10) -- (13);
  \draw[->] (11) -- (14);
  \draw[->] (8)  -- (e0);

  \draw[->] (12) -- (15); 
  \draw[->] (13) -- (15);
  \draw[->] (13) -- (16);
  \draw[->] (14) edge
    node [right=4pt] {\begin{minipage}{0.75cm}\centering{}\tiny{}BII, slack=2\end{minipage}}
    node [sloped,anchor=south,auto=false,rotate=90,yshift=-3pt,minimum width=12pt] (s3) {} (16);
  \draw[->] (12) edge
    node [left=4pt] {\begin{minipage}{0.35cm}\centering{}\tiny{}OI\end{minipage}}
    node [sloped,anchor=south,auto=false,rotate=90,yshift=-7pt,minimum width=12pt] (s4) {} (e1);

  \draw[->] (15) -- (17); 
  \draw[->] (16) -- (17);

  \draw[->] (17) edge
    node [left=4pt] {\begin{minipage}{0.35cm}\centering{}\tiny{}OI\end{minipage}}
    node [sloped,anchor=south,auto=false,rotate=90,minimum width=12pt] (s5) {} (e2); 

  \node at (9)  {\footnotesize $g_0$};
  \node at (10) {\footnotesize $g_4$};
  \node at (15) {\footnotesize $g_1$};
  \node at (16) {\footnotesize $g_3$};
  \node at (17) {\footnotesize $g_2$};
  \node at (e0) {\footnotesize $\vdots$};

  \draw (s0.west) to [in = 265,out = 275,looseness = 0.5] (s0.east); 
  \draw (s1.west) to [in = 265,out = 275,looseness = 0.5] (s1.east); 
  \draw (s2.west) to [in = 265,out = 275,looseness = 0.5] (s2.east); 
  \draw (s3.west) to [in = 265,out = 275,looseness = 0.5] (s3.east); 

  \draw (s4.west) to [in = 85, out = 95, looseness = 0.5] (s4.east); 
  \draw (s5.west) to [in = 85, out = 95, looseness = 0.5] (s5.east); 
\end{tikzpicture}
\caption{A chunk to be moved down.}
\vspace{-0.5em}
\label{fig:down}
\end{figure}

A chunk constructed with Algorithm~\ref{alg:chunk} has flexibilities at all of its interfaces. Thus, even though the individual gates in the chunk cannot be moved legally, a chunk may be moved as a whole. \Cref{fig:down} shows an example chunk. Starting from the initial gate $g_0$, closely-connected gates $g_1, g_2, g_3, g_4$ are added into the chunk in the respective order. The gate $g_1$, for example, cannot be moved up nor down without moving other gates at the same time. In contrast, the gate $g_0$ can be legally moved down, but moving it alone only increases the total number of buffers.

To see how many levels a chunk can be moved and whether the movement reduces the total number of buffers, we define some more properties for chunk interfaces. 

\textit{Moving down:} When a chunk is intended to be moved down, a \emph{slack} is computed at each input interface $(g_c, g_f)$ by
\begin{equation}
    \text{slack}(g_c, g_f) = 
    \begin{cases} 
        & \RD(g_f, g_c) - 1 \text{, if~} |\FO(g_f)| = 1 \\
        & \RD(g_f, g_c) - 2 \text{, otherwise}
    \end{cases}
\end{equation}
The slack of the chunk is the maximum number of levels by which we can move the chunk down, and it is calculated as the minimum slack of all of its input interfaces.
Moreover, $(g_c, g_f)$ is said to be a \emph{beneficial input interface}~(BII) if
\begin{equation}
    \forall g_o \in \FO(g_f), g_o \neq g_c : \RD(g_f, g_o) < \RD(g_f, g_c).
\end{equation}
If a chunk has $x$ BIIs and $y$ OIs with distinct $g_c$, moving the chunk down by $l$ levels eliminates $l \cdot (x - y)$ buffers in total.

\textit{Moving up:} Similarly and conversely, when a chunk is intended to be moved up, a \emph{slack} is computed at each output interface $(g_c, g_f)$ by
\begin{equation}
    \text{slack}(g_c, g_f) = 
    \begin{cases} 
        & \RD(g_c, g_f) - 1 \text{, if~} |\FO(g_c)| = 1 \\
        & \RD(g_c, g_f) - 2 \text{, otherwise}
    \end{cases}
\end{equation}
The slack of the chunk is the minimum slack of all of its output interfaces. When moving up, output interfaces are always beneficial.
If a chunk has $x$ OIs with distinct $g_c$ and $y$ IIs, moving the chunk up by $l$ levels eliminates $l \cdot (x - y)$ buffers in total.

\section{Experimental Results}\label{sec:exp}

\begin{table*}
\centering
\caption{Impact of PI and/or PO balancing and quality of chunked movement.}
\label{tbl:balancing}
\begin{tabular}{lrr rrrrr rrrrr}
\toprule
\multicolumn{3}{c}{\textbf{Unmapped}} &
\multicolumn{10}{c}{\textbf{Balance PIs}}\\

\cmidrule(lr){1-3}
\cmidrule(lr){4-13} 
\multicolumn{3}{c}{} &
\multicolumn{5}{c}{\textbf{Balance POs}} &
\multicolumn{5}{c}{\textbf{Not balance POs}}\\

\cmidrule(lr){4-8}
\cmidrule(lr){9-13}

Bench. & \#gates & Depth &
ASAP & ALAP & Opt. & Depth & \#chunks & 
ASAP & ALAP & Opt. & Depth & \#chunks \\ 
\midrule
 c1908   &   381 &  38  &   3011 &  3296 &  2820 &  64 & 56  &   2605 &  3296 &  2413  & 64 &  56   \\
 c432    &   174 &  44  &   2471 &  2647 &  2220 &  68 & 25  &   2423 &  2635 &  2198  & 70 &  30   \\
 c5315   &  1270 &  33  &   9936 & 11844 &  9457 &  60 & 200 &   6409 & 11402 &  5986  & 59 & 205   \\
 c880    &   300 &  28  &   2577 &  2911 &  2159 &  42 & 44  &   1854 &  2884 &  1501  & 42 &  45   \\
 chkn    &   421 &  28  &   1607 &  1280 &  1241 &  38 & 8   &   1536 &  1280 &  1232  & 38 &   8   \\
 count   &   119 &  18  &    816 &  1004 &   766 &  29 & 31  &    639 &  1004 &   585  & 29 &  31   \\
 dist    &   535 &  16  &   1086 &   814 &   809 &  28 & 3   &   1066 &   814 &   808  & 28 &   3   \\
 in5     &   443 &  19  &   1413 &  1056 &  1042 &  30 & 18  &   1278 &  1056 &  1020  & 30 &  18   \\
 in6     &   370 &  17  &   1184 &   938 &   884 &  23 & 22  &   1002 &   938 &   811  & 23 &  22   \\
 k2      &  1955 &  25  &   5177 &  4570 &  4171 &  43 & 53  &   4512 &  4528 &  3722  & 43 & 139   \\
 m3      &   411 &  13  &    833 &   636 &   620 &  22 & 14  &    761 &   634 &   615  & 22 &  14   \\
 max512  &   713 &  17  &   1399 &  1093 &  1078 &  28 & 3   &   1361 &  1093 &  1070  & 28 &   3   \\
 misex3  &  1532 &  24  &   4181 &  3004 &  2879 &  38 & 16  &   4113 &  3004 &  2883  & 38 &  15   \\
 mlp4    &   462 &  16  &    915 &   668 &   653 &  26 & 9   &    839 &   668 &   647  & 26 &   9   \\
 prom2   &  3477 &  22  &   6855 &  5442 &  5300 &  33 & 59  &   6777 &  5442 &  5298  & 33 &  59   \\
 sqr6    &   138 &  13  &    381 &   246 &   246 &  20 & 8   &    287 &   241 &   229  & 20 &   6   \\
 x1dn    &   152 &  14  &    479 &   561 &   428 &  22 & 16  &    453 &   561 &   399  & 22 &  16   \\
\midrule  
Total    &       &      &  44321 & 42010 & 36773 &     &      &  37915 & 41480 & 31417  &    & \\
Improv.  &       &      &        & 5.2\% & 17.0\%&     &      &        & -9.4\%& 17.1\% &    & \\
Ratio    & & & \bf(1.00) & \it(1.00) &\underline{(1.00)}& & & \bf0.86 & \it 0.99&\underline{0.85}& & \\
\bottomrule
\multicolumn{13}{c}{} \\

%
\toprule
%
%
%

\multicolumn{3}{c}{} &
\multicolumn{10}{c}{\textbf{Not balance PIs}}\\

\cmidrule(lr){4-13} 
\multicolumn{3}{c}{} &
\multicolumn{5}{c}{\textbf{Balance POs}} &
\multicolumn{5}{c}{\textbf{Not balance POs}}\\

\cmidrule(lr){4-8}
\cmidrule(lr){9-13}

Bench. & \#PIs & \#POs &
ASAP & ALAP & Opt. & Depth & \#chunks & 
ASAP & ALAP & Opt. & Depth & \#chunks \\ 
\midrule
 c1908  &  33 &  25 &  3011 &  2910 &  2549 &  62 &  24  &   2605 &  2910 &  2202  & 62 & 44  \\
 c432   &  36 &   7 &  2471 &  1903 &  1689 &  65 &   6  &   2423 &  1891 &  1673  & 65 &  6  \\
 c5315  & 178 & 123 &  9936 &  4520 &  3934 &  56 &  64  &   6409 &  4197 &  3574  & 56 & 64  \\
 c880   &  60 &  26 &  2577 &  1475 &  1306 &  40 &  22  &   1854 &  1448 &  1238  & 40 & 22  \\
 chkn   &  29 &   7 &  1607 &   785 &   720 &  34 &   8  &   1536 &   785 &   715  & 34 &  8  \\
 count  &  35 &  16 &   816 &   343 &   287 &  24 &  15  &    639 &   343 &   286  & 24 & 15  \\
 dist   &   8 &   5 &  1086 &   791 &   762 &  24 &   2  &   1066 &   791 &   761  & 24 &  2  \\
 in5    &  24 &  14 &  1413 &   814 &   762 &  27 &  14  &   1278 &   814 &   746  & 27 & 13  \\
 in6    &  33 &  23 &  1184 &   674 &   627 &  23 &  21  &   1002 &   674 &   621  & 23 & 19  \\
 k2     &  45 &  45 &  5177 &  3854 &  3375 &  37 &  59  &   4512 &  3812 &  3249  & 37 & 56  \\
 m3     &   8 &  16 &   833 &   613 &   576 &  19 &  13  &    761 &   611 &   567  & 19 & 12  \\
 max512 &   9 &   6 &  1399 &  1081 &  1036 &  26 &   3  &   1361 &  1081 &  1028  & 26 &  3  \\
 misex3 &  14 &  14 &  4181 &  2983 &  2815 &  34 &  19  &   4113 &  2983 &  2811  & 34 & 19  \\
 mlp4   &   8 &   8 &   915 &   645 &   609 &  23 &   8  &    839 &   645 &   603  & 23 &  8  \\
 prom2  &   9 &  21 &  6855 &  5435 &  5261 &  33 &  57  &   6777 &  5435 &  5259  & 33 & 57  \\
 sqr6   &   6 &  12 &   381 &   230 &   217 &  17 &   8  &    287 &   225 &   200  & 17 &  6  \\
 x1dn   &  27 &   6 &   479 &   399 &   362 &  19 &   5  &    453 &   399 &   362  & 19 &  5  \\
\midrule  
Total   &     &     & 44321 & 29455 & 26887 &     &      &  37915 & 29044 & 25895  &    & \\
Improv. &     &     &       &33.5\% & 39.3\%&     &      &        & 23.4\%& 31.7\% &    & \\
Ratio   &     &    &\bf1.00 &\it0.70&\underline{0.73}& & &\bf0.86 &\it0.69&\underline{0.70}&  & \\
\bottomrule
\end{tabular}
\end{table*}

\begin{table*}
\centering
\caption{Impact of PI branching and splitting capacity.}
\label{tbl:branchPI}
\setlength{\tabcolsep}{4pt}
\begin{tabular}{lrr rr rr rr rr rr rr}
\toprule
 & & &
\multicolumn{6}{c}{\textbf{Branch PIs}} &
\multicolumn{6}{c}{\textbf{Not branch PIs}}\\

\cmidrule(lr){4-9} 
\cmidrule(lr){10-15} 
&
\multicolumn{2}{c}{$|\FO(i)|$} &
\multicolumn{2}{c}{$s_b = 2$} &
\multicolumn{2}{c}{$s_b = 3$} &
\multicolumn{2}{c}{$s_b = 4$} &
\multicolumn{2}{c}{$s_b = 2$} &
\multicolumn{2}{c}{$s_b = 3$} &
\multicolumn{2}{c}{$s_b = 4$}\\

\cmidrule(lr){2-3}
\cmidrule(lr){4-5}
\cmidrule(lr){6-7}
\cmidrule(lr){8-9}
\cmidrule(lr){10-11}
\cmidrule(lr){12-13}
\cmidrule(lr){14-15}

Bench. & Max. & Avg. &
Opt. & Depth & 
Opt. & Depth & 
Opt. & Depth & 
Opt. & Depth & 
Opt. & Depth & 
Opt. & Depth \\ 
\midrule
 c1908   &   22 &   0.67 &   2624 & 69 &  2202  & 62  &  2073 & 58 &    1245 &  66 &   952 &  61 &   870 & 59  \\
 c432    &    6 &   0.17 &   1932 & 75 &  1673  & 65  &  1512 & 58 &     497 &  74 &   456 &  67 &   423 & 57  \\
 c5315   &   84 &   0.47 &   4818 & 65 &  3574  & 56  &  3087 & 51 &    2668 &  60 &  2043 &  51 &  1796 & 50  \\
 c880    &   11 &   0.18 &   1568 & 50 &  1238  & 40  &  1192 & 40 &     528 &  45 &   437 &  42 &   428 & 42  \\
 chkn    &   42 &   1.45 &   1001 & 40 &   715  & 34  &   602 & 34 &     264 &  34 &   235 &  33 &   229 & 33  \\
 count   &   32 &   0.91 &    373 & 26 &   286  & 24  &   273 & 24 &      72 &  25 &    57 &  23 &    57 & 23  \\
 dist    &   96 &  12.00 &   1093 & 27 &   761  & 24  &   659 & 23 &     392 &  22 &   376 &  22 &   376 & 22  \\
 in5     &   52 &   2.17 &   1019 & 29 &   746  & 27  &   670 & 27 &     351 &  26 &   307 &  26 &   306 & 26  \\
 in6     &   46 &   1.39 &    835 & 25 &   621  & 23  &   544 & 21 &     216 &  20 &   206 &  20 &   205 & 20  \\
 k2      &  152 &   3.38 &   4554 & 39 &  3249  & 37  &  2915 & 36 &    3019 &  38 &  2349 &  35 &  2207 & 35  \\
 m3      &   77 &   9.62 &    861 & 25 &   567  & 19  &   481 & 19 &     306 &  18 &   278 &  18 &   267 & 18  \\
 max512  &  126 &  14.00 &   1475 & 30 &  1028  & 26  &   894 & 24 &     563 &  23 &   541 &  23 &   539 & 23  \\
 misex3  &  144 &  10.29 &   3769 & 43 &  2811  & 34  &  2558 & 34 &    2029 &  33 &  1864 &  33 &  1841 & 34  \\
 mlp4    &   79 &   9.88 &    888 & 25 &   603  & 23  &   514 & 23 &     281 &  22 &   263 &  22 &   263 & 22  \\
 prom2   &  451 &  50.11 &   7369 & 37 &  5259  & 33  &  4568 & 31 &    3813 &  26 &  3405 &  26 &  3346 & 26  \\
 sqr6    &   33 &   5.50 &    292 & 17 &   200  & 17  &   179 & 17 &      92 &  16 &    90 &  16 &    90 & 16  \\
 x1dn    &   15 &   0.56 &    415 & 21 &   362  & 19  &   331 & 19 &     139 &  20 &   124 &  18 &   123 & 18  \\
\midrule
Total    &      &        &  34886 &    &  25895 &     & 23052 &    &   16475 &     & 13983 &     &  13366 &   \\
Ratio    &      &        &   1.35 &    & (1.00) &     &  0.89 &    &    0.64 &     &  0.54 &     &   0.52 &   \\
\bottomrule
\end{tabular}
\end{table*}

In this section, we present experimental results using different combinations of technology assumptions discussed in \Cref{sec:assumptions}.
The irredundant buffer insertion and chunked movement algorithms are implemented in C++-17 as part of the EPFL logic synthesis library \textit{mockturtle}\footnote{Available: \url{github.com/lsils/mockturtle}}~\cite{Soeken18libs}. As discussed in \Cref{subsec:aqfp}, the intrinsic logic gate in the AQFP technology is the majority-$3$ gate, \emph{majority-inverter graphs} (MIGs) \cite{Amaru15} are used as the data structure for (unmapped) networks in our experiments. We use the same initial MIGs as in~\cite{testa2021algebraic} from the MCNC benchmark suite~\cite{mcnc}.

\subsection{Balancing of PIs and POs}
In this section, we use the assumptions that PIs need to be branched ($s_i = 1$) and $s_b = 3$.
\Cref{tbl:balancing} shows the four possible combinations of the assumptions on whether PIs and POs need to be path-balanced. When the POs are not balanced, we do not impose the requirement of modulo-$4$ path lengths either. The block \textbf{Unmapped} lists the benchmark names (Bench.), numbers of majority gates (\#gates), network depths before buffer insertion (Depth), and numbers of PIs (\#PIs) and POs (\#POs). There are five columns in each block, listing the numbers of irredundant buffers after the initial ASAP scheduling (ASAP) and after ALAP (ALAP), the final number of buffers after optimization with chunked movement (Opt.), the depth of the mapped networks (Depth), and the number of chunks (\#chunks). The row Improv. lists the percentage improvements of ALAP and Opt. comparing to ASAP.  The row Ratio lists the ratios of the initial ASAP (\textbf{bold}), ALAP (\textit{italic}), and optimized (\underline{underlined}) buffer counts across different experiments using \textbf{Balance PIs + Balance POs} as the baseline.

The scheduling method (ASAP or ALAP) that leads to fewer buffers is used to obtain the initial depth assignment for chunked movement, decided independently for each benchmark and for each experiment. When PIs need to be balanced (the upper half of \Cref{tbl:balancing}), ALAP do not lead to much improvement, but chunked movement is able to optimize away $17\%$ of the buffers. On the other hand, when PIs do not need to be balanced (the lower half of \Cref{tbl:balancing}), ALAP usually leads to much better results, reducing about $20$-$30\%$ of buffers, and chunked movement further eliminates another $6$-$8\%$. 
Observing the bold ratios, we see that the path-balancing buffers for POs constitute about $14\%$ of the total when ASAP is applied; observing the italic ratios, we see that the path-balancing buffers for PIs constitute about $30\%$ of the total when ALAP is applied. With the chunked-movement-based optimization, the two extremes are balanced, and we see a larger impact of the PI-balancing assumption.

\subsection{Branching of PIs and Splitting Capacity}
In this section, we use the assumptions that neither PIs nor POs need to be balanced (i.e., the last case in \Cref{tbl:balancing}), and we study the impact of branching PIs and the value of buffer's splitting capacity $s_b$. \Cref{tbl:branchPI} shows the number of buffers (Opt.) and the circuit depth (Depth) after optimization using different assumptions on PI branching and $s_b$ value. The two columns under $|\FO(i)|$ show, respectively, the maximum (Max.) and the average (Avg.) fanout size of PIs in each benchmark. Row Ratio lists the ratios of the buffer counts in each experiment comparing to \textbf{Branch PIs}, $s_b = 3$. 

If PIs do not need to be branched, the number of buffers needed is halved. In other words, about half of the buffers are used to branch high-fanout PIs. This phenomenon is even more obvious when the splitting capacity is smaller. Indeed, PIs with high fanout counts are common in many benchmarks, and PI-branching splitters can hardly be eliminated with any optimization.

When not branching PIs, the impact of splitting capacity is relatively minor, with less than $5\%$ difference between $s_b = 3$ and $s_b = 4$. Thus, except for branching PIs, design of high-capacity splitters is not particularly necessary.

\section{Conclusion and Future Work}\label{sec:conclusion}

In conclusion, this paper provides a different viewpoint to the problem of AQFP buffer and splitter insertion. With the linear-time irredundant buffer insertion algorithm presented in \Cref{sec:buf-count} and simple scheduling algorithms discussed in \Cref{subsec:scheduling}, a good starting point can be obtained efficiently. Then, the chunked movement method illustrated in \Cref{subsec:chunk} provides possibility to further minimize the cost and escape from local minima. In \Cref{sec:exp}, experimental results show that the proposed optimization flow with scheduling and chunked movement reduces about $17$-$39\%$ of buffers, depending on the technology assumptions imposed. Moreover, experiments on different assumptions suggest that PI balancing has a greater impact than PO balancing, and that PI branching accounts for half of the inserted buffers. These results motivate future research on the design of AQFP registers and splitters.
For future work, we hope to develop an exact algorithm to find the global optimal, which will allow us to evaluate how good the existing and future-developed heuristics are. We also plan to integrate the proposed buffer optimization with logic synthesis algorithms considering AQFP constraints such as~\cite{testa2021algebraic,marakkalage2021}.

\begin{acks}
This work was supported in part by the EPFL Open Science Fund and by the SNF grant ``Supercool: Design methods and tools for superconducting electronics'', 200021\_1920981. The authors thank Prof. Nobuyuki Yoshikawa and Prof. Christopher L. Ayala from Yokohama National University for their discussion on the AQFP technology.
\end{acks}

\bibliographystyle{ACM-Reference-Format}
\bibliography{iwls}

\end{document}